\documentclass{amsart}
\usepackage[latin1]{inputenc}
\usepackage{amssymb}
\usepackage{lscape}
\usepackage{amsmath}
\usepackage{amsthm}
\usepackage{latexsym} 
\usepackage{calc} 
\usepackage{graphicx}
\usepackage{bm} 
\usepackage{overpic}
\usepackage{hyperref}

\usepackage{exscale}
\usepackage{amsfonts}
\usepackage{amssymb}  
\usepackage[usenames]{color} 

\newtheorem{proposition}{Proposition}

\theoremstyle{definition}

\newtheorem{assumption}{Assumption}

\normalsize

\newcommand{\ind}{1\hspace{-2.1mm}{1}} 
\newcommand{\I}{\mathtt{i}}
\newcommand{\D}{\mathrm{d}}
\newcommand{\E}{\mathrm{e}}

\begin{document}

\title{Convergence of Heston to SVI}
\author{Jim Gatheral}
\address{Bank of America Merrill Lynch and Baruch College, CUNY}
\author{Antoine Jacquier}
\address{Department of Mathematics, Imperial College London and Zeliade Systems, Paris}
\email{ajacquie@imperial.ac.uk.}
\thanks{The authors would like to thank Aleksandar Mijatovi\'c for useful discussions.}
\date{}

\maketitle

\begin{abstract}
In this short note, we  prove by an appropriate change of variables that the SVI implied volatility parameterization presented in~\cite{Gatheral} and the large-time
asymptotic of the Heston implied volatility derived in~\cite{FJM} agree
algebraically, thus confirming a conjecture from~\cite{Gatheral}
as well as providing a simpler expression for the asymptotic implied volatility in the Heston model. We show how this result can help in interpreting SVI parameters.
\end{abstract}

\section{Introduction}

The {\em stochastic volatility inspired} or {\em SVI} parameterization of the implied volatility surface was originally devised at Merrill Lynch in 1999.  This parameterization has two key properties that have led to its subsequent popularity with practitioners: 
\begin{itemize}
\item{For a fixed time to expiry $T$, the implied Black-Scholes variance $\sigma_{BS}^2(k,T)$ is linear in the log-strike $k$ as $|k| \to \infty$ consistent with Roger Lee's moment formula~\cite{Lee}.}
\item{It is relatively easy to fit listed option prices whilst ensuring no calendar spread arbitrage.\footnote{It is seemingly impossible to eliminate the possibility of butterfly arbitrage but this is rarely a problem in practice.}}
\end{itemize}

The result we prove in this note shows that SVI is an exact solution for the implied variance in the Heston model in the limit $T \to \infty$ thus  providing a direct interpretation of the SVI parameters in terms of the parameters of the Heston model.

In Section~\ref{SectionNotation}, we present our notation.  In Section~\ref{sec:saddlepoint}, we motivate the conjecture which we prove in Section~\ref{sec:proof}. 
We conclude in Section~\ref{sec:interpretation} by showing how our result can help us interpret SVI parameters resulting from an SVI fit 
to an empirically observed volatility smile.

\section{Notations}\label{SectionNotation}
From~\cite{Gatheral}, recall that the SVI parameterization for the implied variance reads
\begin{equation}\label{eq:SVI}
\sigma^2_{SVI}\left(x\right)=\frac{\omega_1}{2}\,\left(1+\omega_2\rho x+\sqrt{\left(\omega_2
x+\rho\right)^2+1-\rho^2}\right),\quad\text{for
all } x\in\mathbb{R},
\end{equation}
where $x$ represents the time-scaled log-moneyness, and  consider the Heston model where the stock price process $\left(S_t\right)_{t\geq
0}$ satisfies the following stochastic differential equation:
\begin{align*}
\D S_t & =\sqrt{v_t}S_t \D W_t,\ S_0\in\mathbb{R^*_+}\\
\D v_t & = \kappa\left(\theta-v_t\right)\D t+\sigma\sqrt{v_t}\D Z_t,\ v_0\in\mathbb{R}^*_+\\
\D\langle W,Z\rangle_t & = \rho\,\D t,
\end{align*}
with $\rho\in\left[-1,1\right]$, $\kappa$, $\theta$, $\sigma$ and $v_0$
are strictly positive real numbers satisfying $2\kappa\theta\geq\sigma^2$
(this is the Feller condition ensuring that the process $\left(v_t\right)_{t\geq
0}$ never reaches zero almost surely). We further make the following assumption
as in~\cite{FJM}, under which the Heston asymptotic implied volatility is derived.
\begin{assumption}\label{ass:Coef}
$\kappa-\rho\sigma>0$. 
\end{assumption}
Note that this assumption is usually assumed in the literature, either explicitly or implicitly when assuming a negative correlation $\rho<0$ between the spot 
and the volatility as observed in equity markets. When this condition is not satisfied, the stock price process is still a true martingale, 
but moments greater than one will cease to exist after a certain amount of time, as pointed out in~\cite{Zeliade}, 
which refers to this special case as the \textit{large} correlation regime.
Let us now consider the following choice of SVI parameters 
in terms of the Heston parameters,
\begin{equation}\label{eq:ChangeVar}
\omega_1 :=\frac{4\kappa\theta}{\sigma^2\left(1-\rho^2\right)}\left(\sqrt{\left(2\kappa-\rho\sigma\right)^2+\sigma^2\left(1-\rho^2\right)}-\left(2\kappa-\rho\sigma\right)\right),\quad\text{and}\quad\omega_2 :=\frac{\sigma}{\kappa\theta}.
\end{equation}

Now we know from~\cite{FJM} that the implied
variance in the Heston model in the large time limit $T \to \infty$ takes the following form:
\begin{equation}\label{eq:AsymptH}
\sigma^2_{\infty}\left(x\right)=2\left(2V^*(x)-x+2\left(\ind_{x\in\left(-\theta/2,\bar{\theta}/2\right)}-\ind_{x\in\mathbb{R}\setminus\left(-\theta/2,\bar{\theta}/2\right)}\right)\sqrt{V^*(x)^2-xV^*(x)}\right),\quad\text{for
all } x\in\mathbb{R},
\end{equation}
where $\bar{\theta}:=\kappa\theta/\left(\kappa-\rho\sigma\right)$, and the
function $V^*:\mathbb{R}\to\mathbb{R}_+$ is defined by
\begin{equation}\label{DefOfVStar}
V^*(x):=p^*\left(x\right)x-V\left(p^*\left(x\right)\right),\quad\text{for all }x\in\mathbb{R},
\end{equation}
where
\begin{align*}
V(p) & := \frac{\kappa\theta}{\sigma^2}\Big(\kappa-\rho\sigma
p-d(p)\Big),\quad\text{for all }p\in\left(p_-,p_+\right),\\
d(p) & := \sqrt{\left(\kappa-\rho\sigma p\right)^2 + \sigma^2p\left(1-p^2\right)},\quad\text{for all }p\in\left(p_-,p_+\right),\\
p^*\left(x\right) & := \frac{\sigma-2\kappa\rho+\left(\kappa\theta\rho+x\sigma\right)\eta\left(x^2\sigma^2+2x\kappa\theta\rho\sigma+\kappa^2\theta^2\right)^{-1/2}}{2\sigma\bar{\rho}^2},\quad\text{for all }x\in\mathbb{R},\\
\eta & := \sqrt{4\kappa^2+\sigma^2-4\kappa\rho\sigma},\quad
p_{\pm}:=\left(-2 \kappa \rho+\sigma \pm \sqrt{\sigma^2+4 \kappa^2-4 \kappa \rho\sigma}\right)/\left(2\sigma\bar{\rho}^2\right),
\quad\text{and}\quad \bar{\rho}:=\sqrt{1-\rho^2}.
\end{align*}
Note that in this asymptotic Heston form for the implied volatility, $x$
corresponds to a time-scaled log-moneyness, {\em i.e.} the implied volatility
corresponds to call/put options with strike $S_0\exp\left(xT\right)$, where
$T \geq 0$ represents the maturity of the option.

\section{ The saddle-point condition}\label{sec:saddlepoint}

In this section, we give a non-rigorous motivation for the conjecture in \cite{Gatheral} that the $T \to \infty$ limit of the Heston volatility smile should be SVI.

Consider equation (5.7) on page 60 of \cite{Gatheral:TVS} which relates the implied volatility $\sigma_{BS}(k,T)$ at log-strike $k$ and expiration $T$ to the characteristic function $\phi_T(\cdot)$ of the log-stock price.  We rewrite this equation in the form
\begin{equation}
\int_{-\infty}^\infty\,\frac{\D u}{u^2+\frac{1}{4}}\,\E^{-\I \,u \,k}\phi_T\,
\left(u-\I /2\right)=
\int_{-\infty}^\infty\,\frac{\D u}{u^2+\frac{1}{4}}\,\E^{-\I \,u \,k}\,\E^{-\frac{1}{2}\,\left(u^2+\frac{1}{4}\right)\,\sigma_{BS}^2(k,T)\,T}.
\label{eq:longdates}
\end{equation}
In the limit $T \to \infty$, the Heston characteristic function has the form
\[
\phi_T(u-\I/2) \sim \E^{-\psi(u)\,T}.
\]
Then, as pointed out on page 186 of \cite{Lewis}, we may apply the saddle-point method to both sides in equation (\ref{eq:longdates}) to obtain 
\begin{equation}
\E^{-\I\,k \,\tilde u}\,\frac{\E^{-\psi(\tilde u)\,T}}{\tilde u^2+\frac{1}{4}}\,\sqrt{\frac{2\,\pi}{\psi''(\tilde u)\,T}}
\sim 4\,\exp\left\{-\frac{v\,T}{8}-\frac{k^2}{2\,v\,T}\right\}\,\sqrt{\frac{2\,\pi}{v\,T}},
\label{eq:saddlepoint}
\end{equation}
where $v$ is short-form notation for $\sigma_{BS}^2(k,T)\,T$ and $\tilde u$ is such that
\[
\psi'(\tilde u)=-\I\,\frac{k }{T},
\]
so that $\tilde u$ (which is in general a function of $k$) is a saddle-point, which in the Heston model at least, may be computed explicitly as in Lemma 5.3 of \cite{FJM}.

Defining $k:=x\,T$ and equating the arguments of the exponentials in equation (\ref{eq:saddlepoint}), the dependence on $T$ cancels and we obtain
\begin{equation}
\frac{v(x)}{8}+\frac{x^2}{2\,v(x)}=\psi(\tilde u(x))+\I\,x\,\tilde u(x),
\label{eq:equateexponents}
\end{equation}
where we have reinstated explicit dependence on $x$ for emphasis.

With the help of {\em e.g.} Mathematica, one can verify that in the $T \to \infty$ limit of the Heston model and with the choice (\ref{eq:ChangeVar}) of SVI parameters, 
expression (\ref{eq:SVI}) exactly solves the saddle-point condition (\ref{eq:equateexponents}):
$$
\frac{\sigma_{SVI}^2(x)}{8}+\frac{x^2}{2\,\sigma_{SVI}^2(x)}=\psi(\tilde u(x))+\I\,x\,\tilde u(x).
$$
We are thus led to conjecture that $\sigma^2_{SVI}\left(x\right)=\sigma^2_{\infty}\left(x\right)$ so that the $T \to \infty$ limit of implied variance in the Heston model is SVI.

\section{Main result and proof}\label{sec:proof}
We now state and prove the main result of this note,
\begin{proposition}\label{prop:Result}
Under Assumption~\ref{ass:Coef} and the choice of SVI parameters~\eqref{eq:ChangeVar},
 $\sigma^2_{SVI}\left(x\right)=\sigma^2_{\infty}\left(x\right)$ for all
$x\in\mathbb{R}$.
\end{proposition}
\begin{proof}
Let us now introduce the following  notations: 
$\Delta(x):= \sqrt{\sigma^2 x^2+2\kappa\theta\rho\sigma x+\kappa^2\theta^2}$, where $\eta$ and $\bar{\rho}$ are defined in Section~\ref{SectionNotation}.
Under the change of variables~\eqref{eq:ChangeVar}, the SVI implied variance
takes the form
\begin{equation}\label{eq:SVI2}
\sigma^2_{SVI}\left(x\right)=\frac{2}{\sigma^2\bar{\rho}^2}\Big(\eta-\left(2\kappa-\rho\sigma\right)\Big)\Big(\kappa\theta+\rho\sigma
x+\Delta\left(x\right)\Big),\quad\text{for
all } x\in\mathbb{R}.
\end{equation}
We now move on to simplify the expression for $\sigma^2_{\infty}$ as written
in~\eqref{eq:AsymptH}. We first start by the expression for $V^*(x)$ appearing in~\eqref{eq:AsymptH}.
We have
$$V^*(x) = \frac{A\left(x\right)\Delta\left(x\right)+B(x)\eta}{2\sigma^2\bar{\rho}^2\Delta(x)},$$
with 
$$
A(x):= x\sigma^2-2x\kappa\rho\sigma-2\kappa^2\theta+\kappa\theta\rho\sigma,\quad\text{and}\quad
B(x):= 2x\sigma\kappa\theta\rho+x^2\sigma^2+\kappa^2\theta^2\rho^2+\kappa^2\theta^2\bar{\rho}^2.
$$
Note that $B(x)=\Delta^2\left(x\right)$, so that $V^*\left(x\right)=\left(A(x)+\Delta(x)\eta\right)/\left(2\sigma^2\bar{\rho}^2\right)$.
We further have
\begin{equation}\label{eq:Part1}
2V^*\left(x\right)-x=\frac{A\left(x\right)+\Delta\left(x\right)\eta-x\sigma^2\bar{\rho}^2}{\sigma^2\bar{\rho}^2}=\frac{\Delta\left(x\right)\eta-\left(2\kappa-\rho\sigma\right)\left(\kappa\theta+x\rho\sigma\right)}{\sigma^2\bar{\rho}^2},
\end{equation}
where we use the factorisation $A\left(x\right)-x\sigma^2\bar{\rho}^2=-\left(2\kappa-\rho\sigma\right)\left(\kappa\theta+x\rho\sigma\right)$.

Now, back to~\eqref{eq:AsymptH}, where we denote $\Phi\left(x\right):=V^*(x)^2-xV^*(x)$.
We have
$$\Phi\left(x\right)=\left(\frac{\Delta(x)\eta}{2\sigma^2\bar{\rho}^2}\right)^2+\alpha\left(x\right)\Delta\left(x\right)+\beta\left(x\right),$$
where
$$
\alpha\left(x\right) := -\frac{\eta\left(2\kappa-\rho\sigma\right)\left(\kappa\theta+x\rho\sigma\right)}{2\sigma^4\bar{\rho}^4},\quad\text{and}\quad
\beta\left(x\right) := \frac{1}{4\sigma^4\bar{\rho}^4}\left\{\left(2\kappa-\rho\sigma\right)^2\left(\kappa\theta+x\rho\sigma\right)^2-x^2\sigma^4\bar{\rho}^4\right\}.
$$
We now use the following factorisations: 
\begin{equation}\label{eq:Fact}
\Delta^2\left(x\right)=\left(\kappa\theta+x\rho\sigma\right)^2+x^2\sigma^2\bar{\rho}^2,\quad\text{and}\quad\eta^2=\left(2\kappa-\rho\sigma\right)^2+\sigma^2\bar{\rho}^2,
\end{equation}
so that we can write $\beta\left(x\right)=\left(4\sigma^4\bar{\rho}^4\right)^{-1}\left(\left(2\kappa-\rho\sigma\right)^2\Delta^2\left(x\right)-x^2\sigma^2\bar{\rho}^2\eta^2\right)$
and hence
\begin{align}\label{eq:Part2}
\Phi\left(x\right) 
 & = \frac{1}{4\sigma^4\bar{\rho}^4}\left\{\left[\left(2\kappa-\rho\sigma\right)^2+\sigma^2\bar{\rho}^2\right]\Delta^2\left(x\right)+a\left(x\right)\Delta\left(x\right)+\left(\eta^2-\sigma^2\bar{\rho}^2\right)\left(\Delta^2\left(x\right)-x^2\sigma^2\bar{\rho}^2\right)-x^2\sigma^4\bar{\rho}^4\right\}\nonumber\\
 & = \frac{1}{4\sigma^4\bar{\rho}^4}\left\{\left(2\kappa-\rho\sigma\right)^2\Delta^2\left(x\right)+a\left(x\right)\Delta\left(x\right)+\eta^2\left(\kappa\theta+x\rho\sigma\right)^2\right\}\nonumber\\
 & = \frac{1}{4\sigma^4\bar{\rho}^4}\Big\{\eta\left(\kappa\theta+x\rho\sigma\right)-\left(2\kappa-\rho\sigma\right)\Delta\left(x\right)\Big\}^2,
\end{align}
where, for convenience, we denote $a\left(x\right):=4\sigma^4\bar{\rho}^4\alpha\left(x\right)$.
To complete the proof, we need to take the square root of $\Phi\left(x\right)$,
i.e. we need to study the sign of the expression under the square in~\eqref{eq:Part2}.
Using again \eqref{eq:Fact}, we have
\begin{align*}
&\eta\left(\kappa\theta+x\rho\sigma\right)-\left(2\kappa-\rho\sigma\right)\Delta\left(x\right)\\
 & = \left(\kappa\theta+x\rho\sigma\right)\sqrt{\left(2\kappa-\rho\sigma\right)^2+\sigma^2\bar{\rho}^2}-\left(2\kappa-\rho\sigma\right)\sqrt{\left(\kappa\theta+x\rho\sigma\right)^2+x^2\sigma^2\bar{\rho}^2}\\
 & = \sqrt{\gamma\left(x\right)+\sigma^2\bar{\rho}^2\left(\kappa\theta+x\rho\sigma\right)^2}-\sqrt{\gamma\left(x\right)+x^2\sigma^2\bar{\rho}^2\left(2\kappa-\rho\sigma\right)^2},
 \end{align*}
where $\gamma\left(x\right):=\left(2\kappa-\rho\sigma\right)^2\left(\kappa\theta+x\rho\sigma\right)^2$.
Now, because $\gamma\left(x\right)\geq 0$ for all $x\in\mathbb{R}$, then
the sign of this whole expression is simply given by the sign of the difference
$\psi\left(x\right):=\sigma^2\bar{\rho}^2\left(\kappa\theta+x\rho\sigma\right)^2-x^2\sigma^2\bar{\rho}^2\left(2\kappa-\rho\sigma\right)^2$.
Note further that we actually have $\psi\left(x\right)=\kappa\sigma^2\bar{\rho}^2\left(2x+\theta\right)\left(2x\rho\sigma+\kappa\theta-2\kappa
x\right)$, that this polynomial has exactly two real roots $-\theta/2$ and $\bar{\theta}/2$, and that its second-order coefficient reads $-4\kappa\sigma^2\bar{\rho}^2\left(\kappa-\rho\sigma\right)<0$
under Assumption~\ref{ass:Coef}.
So, plugging~\eqref{eq:Part1} and~\eqref{eq:Part2} into~\eqref{eq:AsymptH}, we exactly obtain \eqref{eq:SVI2} and the proposition
follows.
\end{proof}

\section{Interpretation of SVI parameters}\label{sec:interpretation}

In this section, we use our result in Proposition~\ref{prop:Result} to help interpret the SVI parameters.  From \cite{Gatheral}, the standard SVI parameterization in terms of the log-strike $k$ reads
\begin{equation}
\sigma_{SVI}^2(k)=a+b\,\left\{\tilde \rho\,(k-m)+\sqrt{(k-m)^2+\tilde \sigma^2}\right\}.
\label{eq:SVIclassic}
\end{equation}
Equating (\ref{eq:SVIclassic}) with (\ref{eq:SVI}) and with the parameter choice (\ref{eq:ChangeVar}):
\[
\omega_1 :=\frac{4\kappa\theta}{\sigma^2\left(1-\rho^2\right)}\left(\sqrt{\left(2\kappa-\rho\sigma\right)^2+\sigma^2\left(1-\rho^2\right)}-\left(2\kappa-\rho\sigma\right)\right),
\quad\text{and}\quad\omega_2 :=\frac{\sigma}{\kappa\theta},
\]
we find the following correspondence between SVI parameters and Heston parameters;
\begin{eqnarray}
a&=&\frac{\omega_1}{2}\,(1-\rho^2),\nonumber\\
b&=&\frac{\omega_1\,\omega_2}{2\,T},\nonumber\\
\tilde \rho&=&\rho,\nonumber\\
m&=&-\frac{\rho\,T}{\omega_2},\nonumber\\
\tilde \sigma&=&\frac{\sqrt{1-\rho^2}\,T}{\omega_2}.
\label{eq:SVImap}
\end{eqnarray}
For concreteness, imagine that we are given an SVI fit to the implied volatility smile generated from the Heston model with $T$ very large so that we have the SVI parameters $a$, $b$, $\tilde \rho$, $m$ and $\tilde \sigma$.  Our first observation is that the SVI parameter $\tilde \rho$ is exactly the correlation $\rho$ between changes in instantaneous variance $v$ and changes in the underlying $S$ in the Heston process.  That is, we can read off correlation directly from the orientation of the volatility smile.  In particular, the smile is symmetric when $\rho=0$.

The parameter $b$ gives the angle between the asymptotes of the implied variance smile.  We see from (\ref{eq:SVIclassic}) and (\ref{eq:SVImap}) that the angle between the asymptotes of the total variance smile $\sigma_{SVI}^2(k)\,T$ is constant for large $T$ but that the overall level increases with $T$.

From equation (\ref{eq:SVI}), $\omega_1$ is the at-the-money implied variance $\sigma_{SVI}^2(0,T)$.  
From (\ref{eq:ChangeVar}), in the limit $\sigma \ll \kappa$ , we have
\[
\omega_1=\theta\,\left\{1+\frac{\rho\,\sigma}{\kappa}+O\left(\left(\frac{\sigma}{\kappa}\right)^2\right)
\right\},
\]
so that the at-the-money volatility is given directly by $\theta$ when the volatility of volatility is small.  In the limit $\sigma \gg \kappa$ , we have
\[
\omega_1=\frac{4\, \kappa\,\theta}{\sigma\,(1-\rho)}\,\left\{1-2\,\frac{\kappa}{\sigma}+O\left(\left(\frac{\kappa}{\sigma}\right)^2\right)\right\},
\]
showing that at-the-money volatility decreases as the volatility-of-volatility increases and as the volatility becomes less correlated with the underlying.

Finally, the minimum of the variance smile is attained at $x=-2\,\rho/\omega_2$, providing a simple interpretation of the parameter $\omega_2$. 
In particular, if $\rho=0$, the minimum is exactly the at-the-money point.  The minimum shifts to the upside $x>0$ if $\rho <0$ and to the downside $x<0$ if $\rho > 0$.


\begin{thebibliography}{9}
\bibitem{FJM} M.~Forde, A.~Jacquier, A.~Mijatovic (2009), Asymptotic formulae for implied volatility in the Heston model. Submitted.
\bibitem{Gatheral} J.~Gatheral (2004), A parsimonious arbitrage-free implied volatility parameterization with application to the valuation of volatility derivatives. 
Presentation at Global Derivatives \& Risk Management, Madrid, May 2004, available at 
\href{www.math.nyu.edu/fellows fin math/gatheral/madrid2004.pdf}www.math.nyu.edu/fellows fin math/gatheral/madrid2004.pdf.
\bibitem{Gatheral:TVS} J.~Gatheral (2006), {\em The Volatility Surface: A Practitioner's Guide}, John Wiley \& Sons, Hoboken, NJ.
\bibitem{Lee}Roger~W.~Lee (2004), The Moment Formula for Implied Volatility at Extreme Strikes. {\em Mathematical Finance}, 14, pp 469--480.
\bibitem{Lewis} Alan~L.~Lewis (2000), {\em Option Valuation under Stochastic Volatility with Mathematica Code}. Finance Press, Newport Beach, CA.
\bibitem{Zeliade} Zeliade Systems (2010), Heston 2009. White paper Zeliade. Available at 
\href{http://www.zeliade.com/whitepapers/zwp-0004.pdf}www.zeliade.com/whitepapers/zwp-0004.pdf.
\end{thebibliography}
\end{document}